\documentclass[12pt,english]{article}
\PassOptionsToPackage{natbib=true}{biblatex}
\usepackage[T1]{fontenc}
\usepackage[latin9]{inputenc}
\usepackage{amsmath}
\usepackage{amsthm}
\usepackage{setspace}
\makeatletter
\theoremstyle{plain}
\newtheorem{thm}{\protect\theoremname}
\theoremstyle{plain}
\newtheorem{prop}[thm]{\protect\propositionname}

\usepackage{chenpaper}

\graphicspath{{./exhibits/}}

\usepackage{babel}

\providecommand{\propositionname}{Proposition}
\providecommand{\theoremname}{Theorem}

\makeatother

\usepackage{babel}
\usepackage[style=authoryear,maxcitenames=3,uniquename=false,backend=biber]{biblatex}
\providecommand{\propositionname}{Proposition}
\providecommand{\theoremname}{Theorem}

\begin{document}
\title{Optimal Post-Hoc Theorizing}

\newif\ifanon
\anontrue
\anonfalse 

\ifanon
    \author{}
    \date{}
\else
    \author{{Andrew Y. Chen}\\
    {\normalsize Federal Reserve Board}}
    \date{June 2025\thanks{email:andrew.y.chen@frb.gov. I thank Irene Caracioni for excellent research assistance, and Alejandro Lopez-Lira, Matt Ringgenberg, Mish Velikov, and Tom Zimmermann for helpful comments. The views expressed herein are those of the authors and do not necessarily reflect the position of the Board of Governors of the Federal Reserve or the Federal Reserve System.}}
\fi

\maketitle

\begin{abstract}
\begin{singlespace}
\noindent For many economic questions, the empirical results are not interesting unless they are strong. For these questions, theorizing before the results are known is not always optimal. Instead, the optimal sequencing of theory and empirics trades off a ``Darwinian Learning'' effect from theorizing first with a ``Statistical Learning'' effect from examining the data first. This short paper formalizes the tradeoff in a Bayesian model. In the modern era of mature economic theory and enormous datasets, I argue that \emph{post hoc} theorizing is typically optimal.
\end{singlespace}
\end{abstract}
\vspace{10ex}
\textbf{\color{Black}JEL Classification}: B41, C18, C11

\noindent\textbf{\color{Black}Keywords}: Publication Bias, Machine Learning,  Predictivism vs Accommodation, HARKing
\thispagestyle{empty}\setcounter{page}{0}

\vspace{10ex}

\pagebreak{}

\section{Introduction}

\setcounter{page}{1}

Theories formed after observing empirical results (\emph{post hoc} theories), are viewed with suspicion by social scientists (e.g. \citet{kerr1998harking}; \citet{harvey2017presidential}). Yet some of the most successful theories in all of science were formed this way (e.g. gravity, quantum mechanics).\footnote{\citet{newton1726scholium} even said ``whatever is not deduced from the phenomena... ... have no place in experimental philosophy.''} Consistent with this confusion, the philosophy literature has long debated the merits of \emph{post hoc} vs \emph{a priori} theorizing (\citet{barnes2022prediction})

This paper provides a Bayesian model for understanding this ``paradox.'' It shows \emph{post hoc} theory is clearly suboptimal if the sole goal of research is unbiased empirical results. Given statistics' 100-year obsession with unbiasedness (\citet{efron2001statistical}), it is perhaps unsurprising that \emph{post hoc} theory is viewed suspiciously.

However, the goal of research is typically more than unbiased empirical results. Another ubiquitous goal of research is to find ``a good idea,'' whether the idea is an investment strategy, health intervention, or model of human language. In such settings, statistical bias may matter little, as long as research provides a powerful solution.

If the goal is a ``good idea,'' then the optimal research method trades off a \emph{Darwinian Learning} effect with a \emph{Statistical Learning} effect. Darwinian Learning comes from weeding out bad theories by subjecting them to prediction competitions. Statistical Learning simply comes from theorists improving their ideas after looking at data. If Statistical Learning is stronger than Darwinian Learning, then \emph{post hoc} theorizing is optimal.

In the modern world of enormous datasets and massive computing power, Statistical Learning is becoming more and more powerful. At the same time, the economic sciences have become mature, and Darwinian Learning has arguably run its course. For these reasons, I argue that \emph{post hoc} theorizing is, in most cases, optimal.

For replication code and all previous versions of this paper, see \url{https://github.com/chenandrewy/Post-hoc/}.

\subsection{Related Literature}

My model is an extension of the publication bias models (\citet{hedges1984estimation}; \citet{brodeur2016star}; \citet{andrews2019identification}; \citet{abadie2020statistical}; \citet{chen2020publication}; \citet{jensen2023there}; \citet{kasy2024optimal}). In these papers, it is unclear whether \emph{post hoc} theory is harmful. In fact, the models in these papers exhibit the irrelevance result found in \citet{hempel1966philosophy}; \citet{lakatos1970methodology}; and elsewhere (see Section \ref{sec:ez:irr}). Building on the insights of from the philosophy literature (namely \citealt{maher1988prediction}), I show how heterogeneous theories breaks this irrelevance. 

In the philosophy literature, Maher (\citeyear{maher1988prediction}, \citeyear{maher1990prediction}) and Kahn, Landsburg, and Stockman (\citeyear{kahn1992novel}; \citeyear{kahn1996positive}) (KLS) study \emph{post hoc} theorizing under heterogeneous theories. They document the selection effect that I call Darwinian Learning, and conclude that \emph{a priori} theorizing is optimal, at least in normal scientific settings. Amid the centuries of  debate  (e.g. \citet{leibniz1678letter}; \citet{newton1726scholium}; \citet{keynes1921treatise}), \citet{barnes1996discussion} describes Maher's analysis as ``the closest thing to an illuminating account of predictivism in existence.'' Predictivism is the view that \emph{a priori} theorizing is optimal.

My paper builds on Maher and KLS by showing how there is an offsetting effect to Darwinian Learning, namely Statistical Learning. This effect is ruled out by the assumptions in Maher and KLS. Statistical Learning is perhaps a natural extension of one of \citepos{howson1991maher} criticisms of Maher \citeyearpar{maher1988prediction} and \citeyearpar{maher1990prediction}, though \citet{maher1993discussion} also points out flaws in \citepos{howson1991maher} criticisms.  My paper provides clarity to this debate. Also unlike Howson and Franklin, I show how to connect Maher's and KLS's ideas to the models of publication bias, and the broader statistics literature on large scale inference  (\citet{efron2012large}).

\section{A Very Simple Model of Research}\label{sec:ez}

Idea $i$ is randomly-drawn from a set $\left\{ 1,2,...,N\right\}$, and has quality $\mu_{i}$. $\mu_{i}$ is unknown but researchers can observe the measured quality
\begin{align}
\hat{\mu}_{i} &= \mu_{i} + \varepsilon_{i}
\label{eq:ez:muhat}
\end{align}
where $E\left(\varepsilon_{i}\right) = 0$. $i$ may be a real-world choice for readers (e.g. an investment strategy), in which case $\mu_{i}$ is the realized, quality of $i$ after the research is finished (``post-research''). Or $i$ may be an explanation for some phenomenon (e.g. a model of obesity in adolescents), in which case $\mu_{i}$ is the explanation's fit to the phenomenon, post-research. In either case, higher $\mu_{i}$ is better.

Using theory rules out some ideas:
\begin{align}
    \text{\ensuremath{i} is consistent with theory if }i\in S .
\end{align}
where
\begin{align}
    S \subset \left\{ 1,2,...,N\right\}
    \label{eq:ez:S}.
\end{align}
``Theorizing'' turns $S$ into a selected idea $i^\ast$, and theorizing is either \emph{a priori} or \emph{post hoc}: 
\begin{itemize}
\item \emph{a priori}: the researcher writes down a theory that recommends a selected idea $i^\ast$, which is randomly-selected from $S$. (In this simple model, all ideas are equally consistent with theory.)
\item \emph{post hoc}: the researcher first examines the data (observes $\{\hat{\mu}_1,\hat{\mu}_2,...\hat{\mu}_N\}$). Then she writes down a theory that results in selecting
\begin{align}
    i^\ast =  \arg\max_{i\in S }\hat{\mu}_{i}.\label{eq:given-ihat}
\end{align}
(The researcher chooses the idea with the highest measured quality, subject to the idea being consistent with theory.)
\end{itemize}
In either case, the theory is some math or text that explains why $i^\ast$ is a good idea. In this simple model, the precise nature of the theory is not important, beyond that it argues for selecting $i^\ast$. 

\subsection{Popper's Falsifiability and HARKing}

I assume that idea $i^\ast$ and its supporting theory are eventually re-examined with post-research data via $\mu_{i^\ast}$ (see discussion after Equation \eqref{eq:ez:muhat}).  I also require that $S$ is well-defined, and does not nest all ideas $\left\{1,2,...,N\right\}$.  In other words, I assume theories are falsifiable, in the sense of \citet{Popper1959}. 

One may be concerned that these assumptions are inappropriate for some social sciences. Indeed, \citet{ankel2025economics} provide disturbing evidence that economic theories may be immunized against refutation. If economic theories are indeed, not falsifiable, then they might as well be fairy tales. Whether fairy tales are better told \emph{a priori} or \emph{post hoc} is beyond the scope of this paper.

Perhaps because of \citet{kerr1998harking} (``HARKing: Hypothesizing after the Results are Known''), many researchers equate \emph{post hoc} theorizing with unfalsifiability. However, as seen in this model, constructing theories \emph{post hoc} can be entirely consistent with Popper's notion of science. 

This confusion likely stems from Kerr's loose use of language. The paper has a section titled  ``HARKed Hypotheses Fail Popper's Criterion of Disconfirmability.'' But the text below the title clarifies, ``[a] HARKed hypothesis fails this criterion, at least in a narrow, temporal sense.'' In other words, the text in the section explains that the section title is not necessarily true. In fact, it seems equally reasonable to say that HARKed hypotheses fail Popper's criterion \emph{only} in a narrow, temporal sense.  Errors like these are found throughout \citet{kerr1998harking}. See \citet{rubin2022costs} for a thorough critique.


\subsection{\emph{A Priori} Theorizing is the Unbiased Ideal}

If the sole goal of research is to find an unbiased estimate of idea quality, then \emph{a priori} theorizing achieves this goal.  The expected $\hat{\mu}_{i}$ from \emph{a priori} theorizing satisfies
\begin{align}
E\left(\hat{\mu}_{i}\mid i\in S \right) &= E\left(\mu_{i}\mid i\in S \right).
\label{eq:unbiased}
\end{align}
where $i$ is randomly selected from $S$. In contrast, the expected $\hat{\mu}_{i}$ from \emph{post hoc} theorizing is clearly biased:
\begin{lemma}\label{lem:ez:biased}
    \begin{align}
        E\Bigl(\hat{\mu}_{i}\big|i=\arg\max_{j\in S }\hat{\mu}_{j}\Bigr) & >E\Bigl(\mu_{i}\big|i=\arg\max_{j\in S }\hat{\mu}_{j}\Bigr).
        \label{eq:biased}
    \end{align}            
\end{lemma}
\begin{proof}
    The LHS can be written as 
    \begin{align*}
        &E\left(\mu_i \big| i = \arg\max_{j \in S} \hat{\mu}_j\right) 
        +
        E\left(\varepsilon_i \big| i = \arg\max_{j \in S} \hat{\mu}_j\right) \\        
        &= E\left(\mu_i \big| i = \arg\max_{j \in S} \hat{\mu}_j\right) 
        +
        E\bigg(\varepsilon_i
        \big|
        i \in S,
        \{
         \varepsilon_i 
        > \hat{\mu}_j - \mu_i
        , \quad
         \forall j \in \left( S\setminus\{i\}
         \right)
         \}
         \bigg)
    \end{align*}
The first term is the RHS of Equation \eqref{eq:unbiased}. Thus we just need to show the second term is positive.

The second term is positive because $E\left(\varepsilon_i\big| i \in S\right) = 0$, and because the second condition on $\varepsilon_i$  cuts off the lower tail of the distribution.
\end{proof}

Intuitively, $\hat{\mu}_{i}$ contains both $\mu_{i}$ and measurement error. Selecting on large $\hat{\mu}_{i}$ then selects for positive measurement error, leading to a biased estimate.

The preference for Equation \eqref{eq:unbiased}, and the fear of Equation \eqref{eq:biased}, goes back to \citet{fisher1925statistical}. As described in \citet{efron2001statistical}:
\begin{quote}
    \emph{From the point of view of statistical development, the twentieth century might be labeled ``100 years of unbiasedness.'' Following Fisher's lead, most of our current statistical theory and practice revolves around unbiased or nearly unbiased estimates (particularly MLEs), and tests based on such estimates. The power of this theory has made statistics the dominant interpretational methodology in dozens of fields.}
\end{quote}
Taken with Lemma \ref{lem:ez:biased}, it is no wonder then, that economists are suspicious of \emph{post hoc} theorizing.

\subsection{In Practice, \emph{Post Hoc} Theorizing is Optimal}\label{sec:ez:practical}

In an ideal world, estimates from \emph{a priori} theorizing are all you need. With many, many of these estimates, one  eventually has estimates for every idea, including the best ideas.

But in the real world, consumers and producers of research have limited time. Consumers of research lack the time to read about every idea. Producers of research lack the time to carefully study every idea.

To introduce this real-world limitation, suppose research is restricted to reporting only a single idea, and readers are interested in the idea with the highest quality.

In this case, \emph{post hoc} theorizing is actually optimal. \emph{Post hoc} theorizing uses both the information in theory (Equation \eqref{eq:ez:S}) and the information in the data (Equation \eqref{eq:ez:muhat}), improving its expected quality:
\begin{lemma}\label{lem:ez:post-hoc-opt}
    \begin{align}
        E\bigg(\mu_{i}\big|i=\arg\max_{i'\in S }\hat{\mu}_{i'}\bigg) & >E\left(\mu_{i}\mid i\in S \right).
        \label{eq:ez:post-hoc-opt}
    \end{align}            
\end{lemma} 
\begin{proof}
    The LHS can be written as 
    \begin{align*}
        E\bigg(\mu_i
        \big|
        i \in S,
        \{
         \mu_i 
        > \hat{\mu}_j - \varepsilon_i
        , \quad
         \forall j \in \left( S\setminus\{i\}
         \right)
         \}
         \bigg)
    \end{align*}
The second condition in this expression cuts off the lower tail of the distribution of $\mu_i$. Thus, this expression exceeds the expectation of $\mu_i$ conditioning on $i\in S$ alone, which is the RHS of Equation \eqref{eq:ez:post-hoc-opt}.
\end{proof}

Lemmas \ref{lem:ez:biased} and \ref{lem:ez:post-hoc-opt} are illustrated in Figure \ref{fig:ez}. It simulates 200 selected ideas, with the number of potential ideas $N=100$, $\mu_i \sim \text{Normal}\left(0, 1\right)$, and $\varepsilon_i \sim \text{Normal}\left(0, 1\right)$. \emph{A priori} theorizing leads to less biased estimates, seen in how the dots lie closer to the 45 degree line. However, \emph{post hoc} theorizing leads to higher quality ideas, seen  in how the stars tend to lie toward the right side of the chart.

\begin{figure}[htbp]
    \centering
    \includegraphics[width=\textwidth]{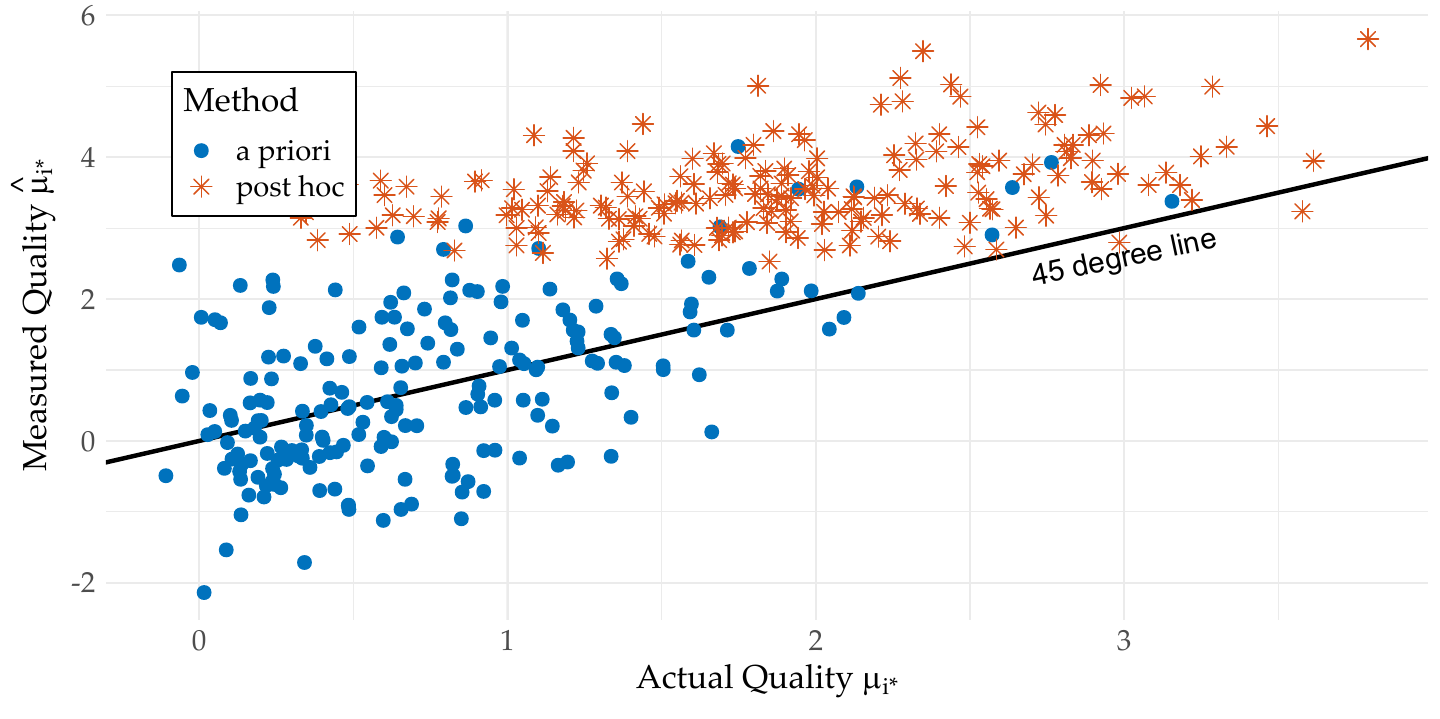}
    \caption{200 Ideas Generated by a Very Simple Model of Research}
    \label{fig:ez}
\end{figure}

The literature on stock market anomalies is an example of Lemma \ref{lem:ez:post-hoc-opt}. Readers are interested in both the magnitude of anomalies, as well as which ones are the strongest.  But assuming that the magnitude meets some minimal standard, readers with limited time will just want to know which anomalies will perform the best in the future. Lemma \ref{lem:ez:post-hoc-opt} shows that, in this case, researchers \emph{should} mine the data, and report what has worked best in the past. This prescription is exactly the reverse of the conventional wisdom, that emphasizes the ``dangers'' of data mining (\citealt{sullivan1999data}; \citealt{harvey2016and}). However, it seems to be in-line with empirical practice, and performs quite well (\citealt{chen2024does}).

Large language models (LLMs) are another example. These models are tuned to perform well on common benchmarks like MMLU (Measuring Massive Multitask Language Understanding) (e.g. \citet{guo2025deepseek}). Thus, the performance on these benchmarks is biased upward, just as in Lemma \ref{lem:ez:biased}. But in practice, this bias is not important, as long as the resulting out-of-sample performance is strong. Tuning improves out-of-sample performance, as seen in Lemma \ref{lem:ez:post-hoc-opt}. 

\subsection{An Irrelevance Result}\label{sec:ez:irr}

In practice, the Fisherian ideal is impossible. Even if all researchers use theory a priori, readers with time constraints are more likely to read the research if the measured effect is large.  This limited attention is arguably the raison d'etre of both peer review (\citet{klamer2002attention}) and publication bias (\citet{chen2022publication})

To model limited attention, suppose \emph{a priori} theory actually involves two steps. First, researchers study all ideas in $S$ and draft up their theories and empirical findings in working papers. However, not all ideas are read.  Due to limited attention, only the idea with the largest measured quality becomes well-known and consumed by the public. The expected quality of this, more realistic, \emph{a priori} theorizing is
\begin{align}
E\Bigl(\mu_{i}\big|i\in S ,i=\arg\max_{i'\in S }\hat{\mu}_{i'}\Bigr) 
& =
E\Bigl(\mu_{i}\big|i=\arg\max_{i'\in S }\hat{\mu}_{i'}\Bigr),
\end{align}
which is exactly the same as the quality of \emph{post hoc} theory (Lemma \ref{lem:ez:post-hoc-opt}).

A similar irrelevance is noted in many works of philosophy (e.g. \citet{hempel1966philosophy}; \citet{lakatos1970methodology}; \citet{rosenkrantz1977inference}; \citet{gardner1982predicting}).  But as noted by \citet{maher1988prediction} and \citet{kahn1996positive}, this irrelevance can be broken if theories are endogenous.

\section{Endogenous, Heterogeneous Theories}\label{sec:het}

Let's make the model richer, with endogenous, heterogeneous theories. This richer model is a generalization of \citet{maher1988prediction} and \citet{kahn1996positive}. Importantly, it allows for an effect I call ``Statistical Learning.'' As in Section \ref{sec:ez:practical}, I assume that the research community has limited time, and is primarily interested in finding ideas with the highest quality.

As before, there are ideas $i\in \{1,2,...,N\}$, measured idea quality $\hat{\mu}_{i}$, and true idea quality $\mu_{i}$. But now  theories come from combining a ``data input'' with a ``theory type.''  

The data input ($\mathcal{D}$ or $\mathcal{O}$) is  known. $\mathcal{D}$ is the case that the data input includes all of the measured effects ($\hat{\mu}_{1},\hat{\mu}_{2},...,\hat{\mu}_{N}$). $\mathcal{O}$ is the case that the theory is given access to none of these effects. \emph{Post hoc} theorizing, then, is represented by $\mathcal{D}$, while \emph{a priori} theorizing is $\mathcal{O}$.

The theory type has a quality $T$ which is unknown. For simplicity, assume the quality is either good (represented by $G$) or bad ($B$). Intuitively, not all theories types are the same, and we may not know how good a particular theory type is.

Combining a theory type with a data input leads to a theory, which in turn provides a recommended idea $i^{\ast}$. As before, $i^\ast$ is a random integer with support $S$, and the theory is some math and/or text that explains why $i^\ast$ is recommended.  But now I'll use conditional probability notation to account for the data input and theory type. For example, $i^{\ast}|G,\mathcal{O}$ is the recommended idea generated by a good theory type and no data (\emph{a priori}).

It's reasonable to think that the good theory type leads to higher quality ideas, \emph{a priori}. This can be formalized by first order stochastic dominance:
\begin{align}
P\left(\mu_{i^{\ast}}>x\mid G,\mathcal{O}\right)
\geq
P\left(\mu_{i^{\ast}}>x\mid B,\mathcal{O}\right),
\quad \forall x\in \mathbb{R}.
\label{eq:given-Gbetter}
\end{align}
For example, one may think that while bad theory types recommend any idea in $S$ with equal probability, good theory types are twice as likely to recommend ideas from the top quartile of $\mu_{i}$ (as compared to the second-to-top quartile). An implication of Equation \eqref{eq:given-Gbetter} is that  good theory types typically lead to higher measured quality $\hat{\mu}_{i^{\ast}}$ than bad theory types.

If theory is done \emph{post hoc}, researchers examine measured qualities $\hat{\mu}_1,\hat{\mu}_2,...,\hat{\mu}_N$, as well as the theory type, to construct a theory that selects idea $i^{\ast}|T,\mathcal{D}$. I allow $i^{\ast}|T,\mathcal{D}$ to be general, but assume the following restriction:
\begin{align}
P\left(i^{\ast}=\arg\max_{i\in S }\hat{\mu}_{i}\mid B,\mathcal{D}\right) & =1.0,
\label{eq:endo:bad-post-hoc}
\end{align}
that is, using bad type theories always lead researchers to select the idea with the strongest measured quality (provided the idea is consistent with \emph{some} theory).  This assumption can be thought of as bad theory types being unable to distinguish between ideas in $S$, and Bayesian researchers who optimize on the posterior mean based on this information and $\hat{\mu}_{i}$ (see \citealt{chen2025high}). 

After $i^{\ast}$ is chosen, readers decide if they are interested in the theory and idea. Assume readers are uninterested unless 
\begin{align}\label{eq:hurdle}
\hat{\mu}_{i^{\ast}} & >h,
\end{align}
where $h$ is some kind of economic and/or statistical hurdle. Only theories and ideas readers are interested in are published. This assumption follows the econometric literature on publication bias (\citet{andrews2019identification}).

\subsection{Darwinian Learning}

An immediate implication of heterogeneous theories is heterogeneous measured quality:
\begin{lemma}
    \label{lem:darwin}
    \begin{align}
        P\left(\hat{\mu}_{i^{\ast}}>h|G,\mathcal{O}\right) 
        &>
        P\left(\hat{\mu}_{i^{\ast}}>h|B,\mathcal{O}\right)
    \end{align}
\end{lemma}
\begin{proof}
    Since $\varepsilon_{i}$ is i.i.d., adding it to $\mu_{i^{\ast}}$ preserves first-order stochastic dominance.
\end{proof}

Lemma \ref{lem:darwin} provides an alternative way to think about the \citet{chen2022peer} (CLZ) ``peer review vs data mining'' experiment. CLZ compare stock trading ideas from peer review to data-mined trading ideas, using post-publication returns. If we call the post-publication returns $\hat{\mu}_{i^\ast}$, neither the peer-reviewed nor data-mined ideas had access to this data, so $\mathcal{O}$ holds for both groups of ideas. Then, one can think of peer-reviewed ideas as $i^\ast|T,\mathcal{O}$, since we do not know if the theory type is $G$ or $B$. In contrast, we can think of the data-mined ideas $i^\ast |B,\mathcal{O}$. As powerfully demonstrated by \citet{novy2025ai}, anyone can add text to these ideas and call it a theory. 

From this framing, CLZ's empirical results are a test of whether $G$ theory types exist. If $G$ theory types comprise a significant fraction of the theories in the CLZ sample, then Lemma \ref{lem:darwin} implies that the published strategies have higher $\hat{\mu_{i^\ast}}$. Unfortunately, CLZ find that published strategies fail to outperform, implying that $G$ theories are rare.

The CLZ experiment illustrates the Darwinian selection of theories. If we force theorists to announce their ideas before looking at the data, then the bad theory types cannot hide  behind data mining.  This intuition helps justify the belief that \emph{a priori} theorizing provides ``discipline'' and that \emph{post hoc} theorizing is ``too easy.''  The following proposition formalizes this idea:
\begin{prop}\label{prop:darwin}
{[}Darwinian Selection of Theories{]} 
\begin{align*}
P\left(G|\mathcal{O},\hat{\mu}_{i^{\ast}}>h\right)-P\left(G|\mathcal{\mathcal{D}},\hat{\mu}_{i^{\ast}}>h\right) & >0
\end{align*}
\end{prop}

\begin{proof}
Apply Bayes rule to the LHS and simplify to yield
\begin{align*}
\frac{P\left(\hat{\mu}_{i^{\ast}}>h|G,\mathcal{O}\right)}{P\left(\hat{\mu}_{i^{\ast}}>h|B,\mathcal{O}\right)} & >\frac{P\left(\hat{\mu}_{i^{\ast}}>h|G,\mathcal{\mathcal{D}}\right)}{P\left(\hat{\mu}_{i^{\ast}}>h|B,\mathcal{\mathcal{D}}\right)}
\end{align*}
Lemma \ref{lem:darwin} shows that the LHS is greater than 1.0. But since bad theories always select the largest
$\hat{\mu}_{i}$ \emph{post hoc} (Equation \eqref{eq:endo:bad-post-hoc}), the RHS is at most 1.0. 
\end{proof}

Proposition \ref{prop:darwin} is illustrated in Figure \ref{fig:darwinian_selection}. It shows histograms generated by parameters deliberately chosen to highlight the power of Darwinian selection.

\begin{figure}[h]
    \def\tempwidth{1.0\textwidth}

    \centering
    \begin{subfigure}[b]{\textwidth}
        \centering
        \caption{\emph{A Priori} Theorizing}
        \includegraphics[width=\tempwidth]{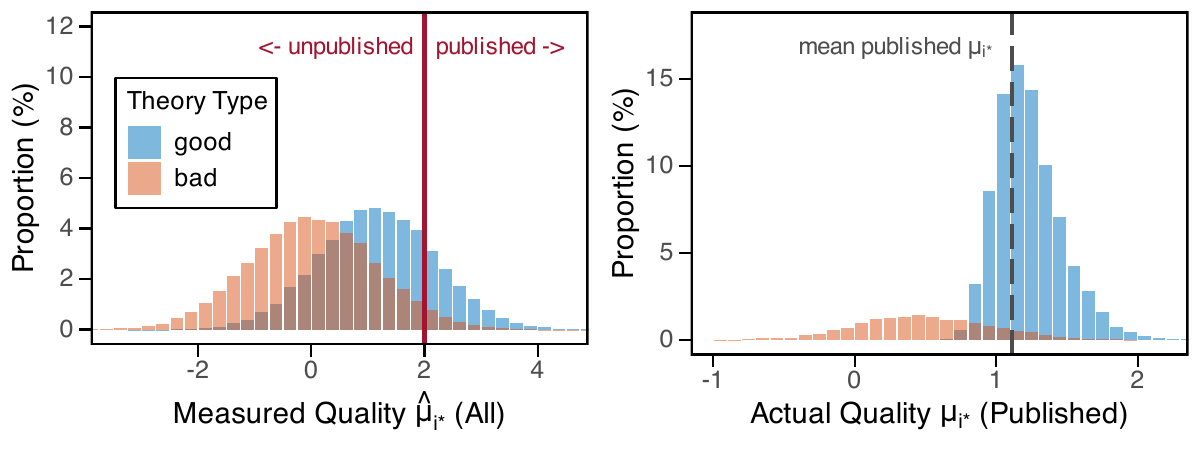}
    \end{subfigure}
    
    \vspace{1em}
    
    \begin{subfigure}[b]{\textwidth}
        \centering
        \caption{\emph{Post Hoc} Theorizing}
        \includegraphics[width=\tempwidth]{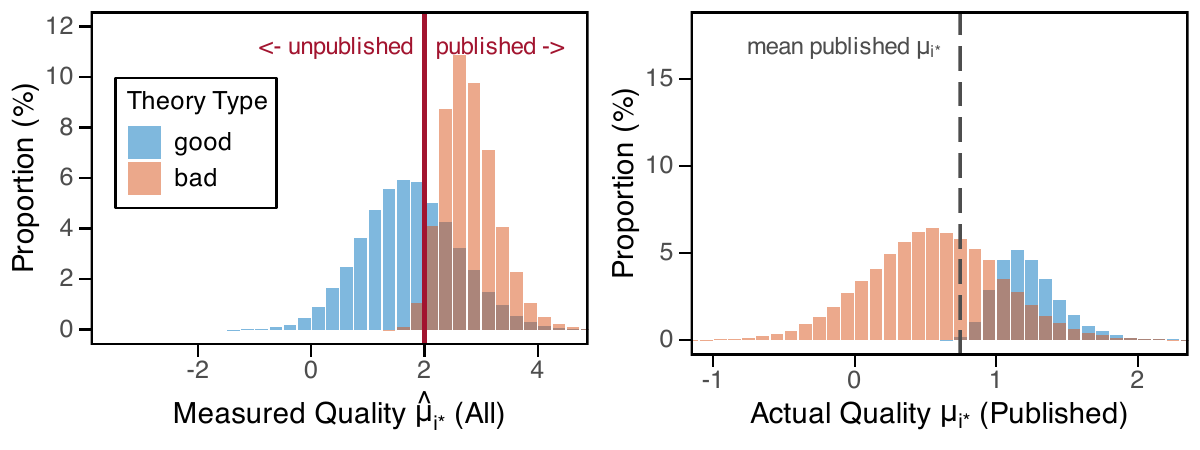}
    \end{subfigure}
    \caption{
        \textbf{Darwinian Selection Illustration.} Histograms show all selected ideas (All) or those that meet the hurdle $h=2.0$ (Published). Number of ideas $N=100$, actual quality $\mu_i \sim N(0, 0.5^2)$, noise $\varepsilon_i \sim \text{Normal}(0, 1)$. Prob of a good type is 50\%. \emph{A priori}, bad types equally weight all ideas, while good types equally weight the two best. \emph{Post hoc}, researchers select the idea with the highest $\hat{\mu}_i$ with positive \emph{a priori} weight.  
    }
    \label{fig:darwinian_selection}
\end{figure}

Under \emph{a priori} theorizing, published ideas mostly come from the good theory types (Panel (a), left). Naturally, good theory types are better at separating good ideas from bad ones, \emph{a priori}. \emph{Post hoc}, published ideas largely come from the bad theory types (Panel (b), left). This happens because bad theory types lead researchers to check far more ideas for the highest measure quality, as these bad theories cannot discriminate among ideas. As a result, bad theory types are more likely to lead to publication, despite having lower actual quality. The final result is that \emph{a priori} theorizing leads to published ideas with higher actual quality (vertical dashed lines).

Proposition \ref{prop:darwin} captures the key insight of Maher (\citeyear{maher1988prediction}; \citeyear{maher1990prediction}) and Kahn, Landsburg, and Stockman (\citeyear{kahn1992novel}, \citeyear{kahn1996positive}). If theories are heterogeneous, then forcing theorists to announce their ideas before looking at the data helps eliminate bad theories, as in Darwinian selection.  In Maher's terminology, a theory  is a ``method,'' and the theory type is ``reliability,'' but the idea is the same. 


Maher and KLS push further. They claim that, not only does \emph{a priori} theorizing produce Darwinian selection, but that the resulting hypotheses are more likely to be true. The analogue here is that $\mathcal{O}$ implies not only that $G$ is more likely, but that $\mu_{i^\ast}$ is higher. We'll see that this conclusion is not necessarily true.\footnote{\citet{barnes1996discussion} revisits Maher (1988, 1990, 1993) and does not go further. His Eq (4) stops here, and considers more deeply the terms in the Bayes rule version of $P\left(G|\mathcal{O},\hat{\mu}_{i^{\ast}}>h\right)$. 
} 

An interesting feature of Proposition \ref{prop:darwin} is that it shows a virtue of publication bias. While requiring $\hat{\mu}_{i^{\ast}}>h$ leads to biased estimates, it helps weed out bad theories types. This result is closely analogous to Lemma \ref{lem:ez:post-hoc-opt}.

\subsection{Optimal Post-Hoc Theory}

Research is not only interested in finding good theory types, but also good ideas. In fact, one can argue that finding good ideas is the ultimate goal. 

Whether \emph{post hoc} theory helps or hurts for finding good ideas is characterized by the following proposition:
\begin{prop}\label{prop:optimal_post_hoc}
{[}Optimal Post Hoc Theory{]}
\begin{align}
E\left(\mu_{i^{\ast}}|\mathcal{\mathcal{D}},\hat{\mu}_{i^{\ast}}>h\right) & >E\left(\mu_{i^{\ast}}|\mathcal{O},\hat{\mu}_{i^{\ast}}>h\right)
\end{align}
if and only if 
\begin{align}
\text{\ensuremath{\left[\text{Statistical Learning}\right]}} & >\left[\text{Darwinian Learning}\right]
\end{align}
where 
\begin{align}
\left[\text{Darwinian Learning}\right] 
& \equiv
    \big[
        P\left(G|\mathcal{O},\hat{\mu}_{i^{\ast}}>h\right)-
    P\left(G|\mathcal{D},\hat{\mu}_{i^{\ast}}>h\right)
    \big]
    \notag
    \\
& \quad\times
\big[
    E\left(\mu_{i^{\ast}}|G,\mathcal{O},\hat{\mu}_{i^{\ast}}>h\right) - 
    E\left(\mu_{i^{\ast}}|B,\mathcal{O},\hat{\mu}_{i^{\ast}}>h\right)
\big] 
\label{eq:endo:darwinian_learning} \\
\left[\text{Statistical Learning}\right] & 
\equiv 
P\left(G|\mathcal{D}, \hat{\mu}_{i^{\ast}}>h\right)\left[
    E\left(\mu_{i^{\ast}}|G,\mathcal{D}, \hat{\mu}_{i^{\ast}}>h\right) - 
    E\left(\mu_{i^{\ast}}|G,\mathcal{O}, \hat{\mu}_{i^{\ast}}>h\right)
\right] \notag \\ 
& \quad +
P\left(B|\mathcal{D}, \hat{\mu}_{i^{\ast}}>h\right)\left[
    E\left(\mu_{i^{\ast}}|B,\mathcal{D}, \hat{\mu}_{i^{\ast}}>h\right) - 
    E\left(\mu_{i^{\ast}}|B,\mathcal{O}, \hat{\mu}_{i^{\ast}}>h\right)
\right]
\label{eq:endo:statistical_learning}
\end{align}
\end{prop}
The proof is in Appendix \ref{sec:app:proof1}.

The proposition says that whether \emph{post hoc} or \emph{a priori}  theorizing leads to better ideas depends on the relative size of two effects:
\begin{enumerate}
    \item Darwinian Learning: This measures the ultimate effect of Darwinian selection (Proposition \ref{prop:darwin}), which occurs when researchers are forced to predict without data ($\mathcal{O}$).  Intuitively, Darwinian selection improves ideas only to the extent that $G$ theory types find higher $\mu_{i^\ast}$ compared to $B$ theory types (second line of Equation \eqref{eq:endo:darwinian_learning}). 
    \item Statistical Learning: This measures how idea quality $\mu_{i^{\ast}}$ improves when the researcher has access to more data ($\mathcal{D}$). Just as how a Bayesian improves her inferences with new evidence, theorists develop higher quality ideas with access to data. 
\end{enumerate}
Naturally, if Statistical Learning exceeds Darwinian Learning, then it's often better to look at the data---i.e. \emph{post hoc} theory may be optimal.

There is no hard and fast rule for which effect is larger. There are certainly settings where Statistical Learning is miniscule (e.g. when the data is extremely noisy). And there are certainly settings where Darwinian Learning is ineffective (e.g. when all theory types are the same). 

Similarly, there are contradictory historical examples. Mendeleev's prediction of elements is a shockingly impressive example of \emph{a priori} theorizing. But Planck's law of radiation is a shockingly impressive example of \emph{post hoc} theorizing. Proposition \ref{prop:optimal_post_hoc} provides a way to understand these seemingly contradictory phenomena.


\subsection{When is \emph{post hoc} theorizing optimal?}\label{sec:het-figures}

If theories are homogenous in quality, then there is no Darwinian Learning, and thus Proposition \ref{prop:optimal_post_hoc} implies that \emph{post hoc} theory is optimal.

Figure \ref{fig:prop2-het} illustrates this phenomenon, by examining many variations of the model from Figure \ref{fig:darwinian_selection}. In Figure \ref{fig:prop2-het}, theory types were extremely heterogeneous: bad theory types cannot eliminate any ideas, while good theory types eliminate the worst 98\% of ideas. This extreme-heterogeneity model is shown in the right most markers of Figure \ref{fig:prop2-het}. For this model, the improvement from \emph{post hoc} theory is a negative 30\%: i.e. published ideas have 30\% lower quality under \emph{post hoc} theory (top panel). Correspondingly, Darwinian Learning is very large, and far exceeds Statistical Learning (bottom panel).

However, reducing the heterogeneity of theories leads to \emph{post hoc} theory being optimal. Moving from right to left in Figure \ref{fig:prop2-het}, the improvement from \emph{post hoc} theory turns positive once good theory types can eliminate the worst 75\% of ideas. Here, Statistical Learning is exactly equal to Darwinian Learning (bottom panel). For models with any less heterogeneity, \emph{post hoc} theory is optimal. 

\begin{figure}
    \centering
    \includegraphics[width=0.8\textwidth]{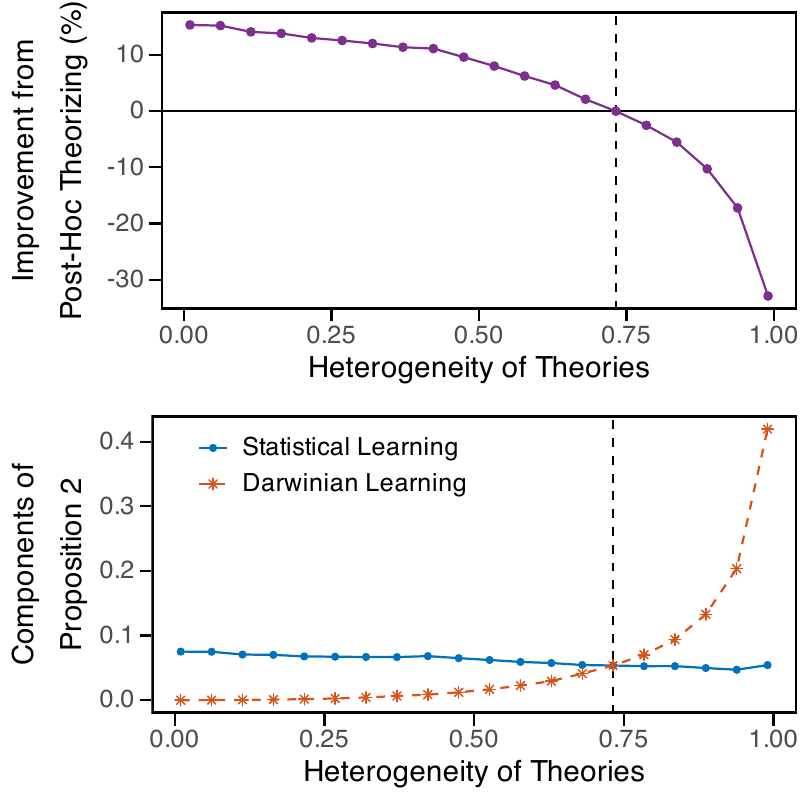}
    \caption{\textbf{Optimal Theorizing vs Heterogeneity of Theories.} Each marker is one model. `Improvement from Post-Hoc Theorizing' is $E\left(\mu_{i^{\ast}}|\mathcal{D}, \hat{\mu}_{i^{\ast}}>h\right)/E\left(\mu_{i^{\ast}}|\mathcal{O}, \hat{\mu}_{i^{\ast}}>h\right)-1$ (see Proposition \ref{prop:optimal_post_hoc}). `Heterogeneity of Theories' is the share of ideas that good theories can eliminate \emph{a priori}. Otherwise, the model is the same as in Figure \ref{fig:darwinian_selection}, in which bad theories cannot eliminate any ideas.}
    \label{fig:prop2-het}
\end{figure}

\subsubsection{Large Datasets and Optimal Theorizing}\label{sec:het-bigdata}

Another implication of Proposition \ref{prop:optimal_post_hoc} is that larger datasets tend to imply \emph{post hoc} theory is optimal. Naturally, larger datasets imply more Statistical Learning. 

To model this, one can think of measured quality $\hat{\mu}_i$ as a t-statistic, in which case a large dataset implies high $\Var(\hat{\mu}_i)$.  Intuitively, as the sample size increases, so does the probability of finding statistically-significant t-stats (\citealt{abadie2020statistical}). 

To formalize this interpretation, suppose that underlying Equation \eqref{eq:ez:muhat} is a panel data model:
\begin{align}\label{eq:panel-data}
    x_{i,j} &= \chi_i + e_{i,j}, \quad j=1,2,\cdots,M \\
     \Var(e_{i,j}) &= \sigma^2,
\end{align}
where $M$ is the number of observations for idea $i$. Moreover, suppose we fix the hurdle for readers' interest at $h=2.0$ (see Equation \eqref{eq:hurdle}). Then a natural way to map Equation \eqref{eq:panel-data} to Equation \eqref{eq:ez:muhat} is to define $\hat{\mu}_i$ as the t-statistic for $\chi_i$:
\begin{align}\label{eq:NTscaling}
    \hat{\mu}_{i,t} &= \frac{\bar{x}_i}{\sigma_i} \sqrt{M}
    = \underbrace{\frac{\sqrt{M}}{\sigma_i} \chi_i}_{\mu_i}
    + \underbrace{\frac{\sqrt{M}}{\sigma_i}\bar{e}_i}_{\varepsilon_i}, \\
    \label{eq:NTscaling2}
    \Var_{i}\left(\varepsilon_i\right) &= 1.0,
\end{align}
where $\Var_{i}(\varepsilon_i)$ is the variance holding fixed the idea, and Equation \eqref{eq:NTscaling2} assumes that the central limit theorem holds and $\sigma^2$ is observed.  Thus, in this setting, the standard deviation of $\hat{\mu}_i$ is increasing in the sample size $M$. 

Figure \ref{fig:prop2-bigdata} illustrates how large datasets affect optimal theorizing, interpreted through the panel data model (Equations \eqref{eq:panel-data}-\eqref{eq:NTscaling2}). It revisits the model from Figure \ref{fig:darwinian_selection}, but examines alternative choices for the variance of $\mu_i$. The x-axis plots $\sqrt{\Var(\hat{\mu}_i)}$, which can be interpreted as either the dispersion of t-statistics or  a measure of the sample size.  

\begin{figure}
    \centering
    \includegraphics[width=0.8\textwidth]{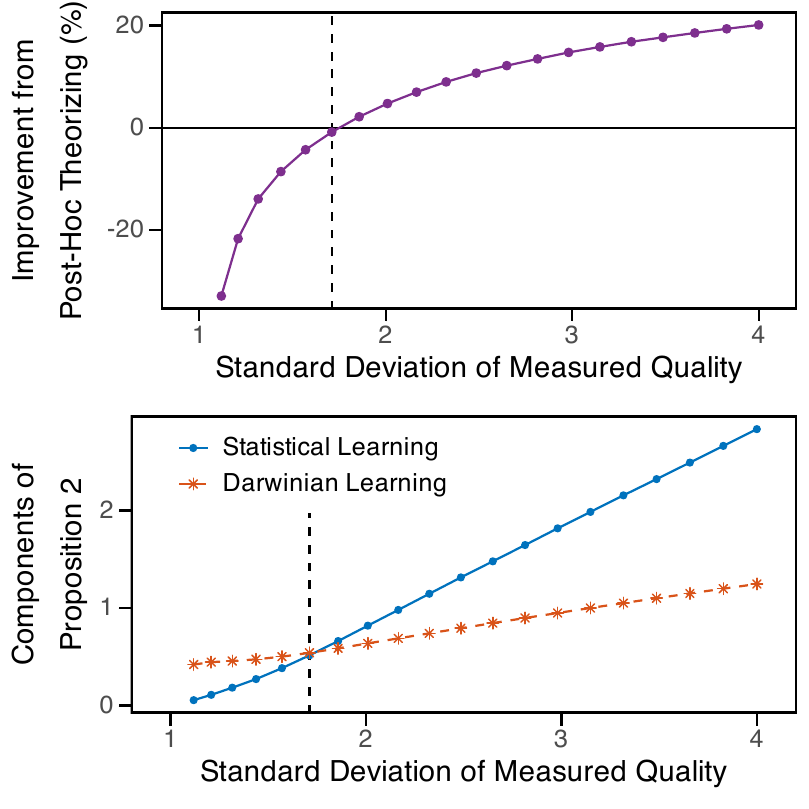}
    \caption{\textbf{Optimal Theorizing vs Sample Size.} Each marker is one model with a different $\sqrt{\Var(\mu_i)}$. The rest of the model is the same as in Figure \ref{fig:darwinian_selection}. `Improvement from Post-Hoc Theorizing' is $E\left(\mu_{i^{\ast}}|\mathcal{D}, \hat{\mu}_{i^{\ast}}>h\right)/E\left(\mu_{i^{\ast}}|\mathcal{O}, \hat{\mu}_{i^{\ast}}>h\right)-1$. `Standard Deviation of Measured Quality' is $\sqrt{\Var(\hat{\mu}_i)}$, which can be interpreted as the dispersion of t-statistics or a measure of sample size (Equations \eqref{eq:panel-data}-\eqref{eq:NTscaling2}).}
    \label{fig:prop2-bigdata}
\end{figure}

The left-most markers correspond to the model from Figure \ref{fig:darwinian_selection}. This model was selected to illustrate the power of Darwinian selection. Thus, $\Var(\hat{\mu}_i)$ is close to 1.0, indicating the measured quality is close to the null distribution, and noise dominates the data. Thus, Statistical Learning is small, and \emph{a priori} theorizing is optimal.

But as $\Var(\hat{\mu}_i)$ increases, so does the amount of signal, holding fixed $\Var(\varepsilon_i)$ at 1.0. The amount of Statistical Learning then increases, and \emph{post hoc} theorizing, starts to become optimal at approximately $\sqrt{\Var(\mu_i)}=1.75$. 

$\sqrt{\Var(\mu_i)}=1.75$ relatively small. For comparison, \citet{chen2020publication} and \citet{jensen2023there} estimate $\sqrt{\Var(\mu_i)}\approx 3.0$ for empirical asset pricing (see discussion in \citealt{chen2022publication}). For settings like this, where $\hat{\mu}_i$ provides a strong signal about the underlying $\mu_i$, Statistical Learning most likely exceeds Darwinian Learning, and thus \emph{post hoc} theorizing is typically optimal.

\subsubsection{Optimal Theorizing in Modern Economics}\label{sec:het-econcomment}

As a field of research matures, institutions arise that standardize the many aspects of research, including the peer review process, the statistical analysis, and  theory. It is reasonable to think, then, that mature fields have theories that are relatively homogeneous in quality. In fact, homogeneous theory quality is a reasonable definition of a mature field.

Economics is arguably mature. Before the 1950s, there was wild variety in the way that economists theorized. But theory began to solidify with the contributions of Arrow and Samuelson. And though behavioral economics has risen in popularity in recent decades, and the 2008 financial crisis brought on significant criticism of economic models, the basic structure of theory has been largely stable since the 1980s. It is thus reasonable to think that economic theories are fairly homogeneous in quality, and that Darwinian Learning is small. 

At the same time, the modern era has seen the rise of huge datasets and enormous computing power. As discussed in Section \ref{sec:het-bigdata}, this implies that standardized measures of idea quality are dispersed, and thus Statistical Learning is large. 

Taken together, these arguments imply that \emph{post hoc} theory is typically optimal in the modern era of economics.

This argument has some surprising implications. Pre-analysis plans should \emph{not} be followed.  Journals should \emph{favor}  theories that accommodate the data, \emph{post hoc}. At least, these are the prescriptions for a literature that focuses on finding the best ideas, and places less emphasis on unbiasedness. 

While it may feel uncomfortable to favor results over unbiasedness, this is precisely the approach taken by the computer science literature. Following this practical route, computer science has essentially taken over machine learning, which could have been the territory of statisticians. Perhaps the maturation of statistical theories, as well as the rise of big data, tilted the balance in favor of \emph{post hoc} theorizing, and thus the dominance of computer scientists.

\section{Conclusion}

This paper presents a framework for understanding several questions about the scientific method: Why is \emph{post hoc} theorizing viewed as a problem? How do we square this problem with highly-successful \emph{post hoc} theories? Does the classical view of \emph{post hoc} theory still hold up in the modern era of big data? 

The framework shows that the distrust of \emph{post hoc} theorizing is to a significant extent a relic of idealized, pre-modern statistics. With practical constraints on researchers' time, and a focus on results over unbiasedness, \emph{a priori} theorizing is not always superior. Instead, there is a trade-off between  Darwinian Learning, which comes from forcing theorists into prediction contests, and Statistical Learning, which arises as researchers learn from data. With  modern datasets and computing power, Statistical Learning is clearly very significant. At the same time, it is unclear that Darwinian Learning still matters, in a world of mature theories.

A caveat is that \emph{a priori} theorizing has benefits that are omitted from my analysis. Most important, \citet{barnes2008paradox} points out that prediction contests provide an accessible, democratic way to establish what is good science. The main alternative is the peer review process, which is inscrutable to outsiders, and can potentially be abused.\footnote{Additionally, KLS argue that the choice of \emph{a priori} vs \emph{post hoc} theorizing may be endogenous, which can lead to additional selection effects, over and above Proposition \ref{prop:optimal_post_hoc}. However, the basic logic that \emph{a priori} theorizing helps through inducing selection is still captured by Proposition \ref{prop:optimal_post_hoc}.}

A second caveat is that none of this analysis matters if economic theories are not, eventually, tested with post-research data. If economic theories are really not falsifiable, then the value of \emph{a priori} and \emph{post hoc} theorizing is an unscientific question, and thus beyond the scope of this paper.

\begin{appendices}
\section{Proof of Proposition \ref{prop:optimal_post_hoc}}\label{sec:app:proof1}
\begin{proof}
    For ease of notation, let $\tilde{E}$ be the expectation operator conditioned on $\hat{\mu}_{i^{\ast}}>h$ and define $\tilde{P}$ similarly. Also define conditioning on $I\in \left\{ \mathcal{O},\mathcal{D}\right\}$ and $I'\in \left\{ \mathcal{O},\mathcal{D}\right\}$ as
    \begin{align}
        \tilde{E}\left\{ \tilde{E}\left(\mu|T,I\right)|I'\right\} 
            &\equiv \tilde{P}\left(G|I'\right)\tilde{E}\left(\mu|G,I\right)+\tilde{P}\left(B|I'\right)\tilde{E}\left(\mu|B,I\right),
    \end{align}
    which can be rewritten as
    \begin{align}\label{eq:app-proof-2}
        \tilde{E}\left\{ \tilde{E}\left(\mu|T,I\right)|I'\right\} 
            &\equiv \tilde{E}\left(\mu|B,I\right) 
            + \tilde{P}\left(G|I'\right)
            \left\{ \tilde{E}\left(\mu|G,I\right)-\tilde{E}\left(\mu|B,I\right)\right\}. 
    \end{align}

    The expected quality from \emph{a priori} theory can be written as
    \begin{align}
    \tilde{E}\left\{ \mu_{i^{\ast}}|\mathcal{O}\right\} 	
    &= \tilde{E}\left\{ \tilde{E}\left[\mu_{i^{\ast}}|T,\mathcal{O}\right]|\mathcal{O}\right\} 
    - \tilde{E}\left\{ 
        \tilde{E}\left[\mu_{i^{\ast}}|T,\mathcal{O}\right]|\mathcal{D}
    \right\} \notag \\
    &\quad + \tilde{E}\left\{ \tilde{E}\left[\mu_{i^{\ast}}|T,\mathcal{O}\right]|\mathcal{D}\right\}, 
    \end{align}  
    where the first term uses iterated expectations and the last two terms sum to zero. Thus the expected quality difference of \emph{a priori} vs \emph{post hoc} theory is
    \begin{align}
        \tilde{E}\left\{ \mu_{i^{\ast}}|\mathcal{O}\right\} -\tilde{E}\left\{ \mu_{i^{\ast}}|\mathcal{D}\right\} 	
            &=\tilde{E}\left\{ \tilde{E}\left[\mu_{i^{\ast}}|T,\mathcal{O}\right]|\mathcal{O}\right\} -\tilde{E}\left\{ \tilde{E}\left[\mu_{i^{\ast}}|T,\mathcal{O}\right]|\mathcal{D}\right\} \notag \\
            &\quad-\left\{
                \tilde{E}\left[\mu_{i^{\ast}}|\mathcal{D}\right]
                -\tilde{E}\left\{\left[\mu_{i^{\ast}}|T,\mathcal{O}\right]|\mathcal{D}\right\}
            \right\}
            \label{eq:optimal_post_hoc_1}
    \end{align}
    The second line of the RHS of \eqref{eq:optimal_post_hoc_1} is $\left[\text{Statistical Learning}\right]$ (just apply iterated expectations to $\tilde{E}\left[\mu_{i^{\ast}}|\mathcal{D}\right]$).

    The first line of the RHS of \eqref{eq:optimal_post_hoc_1} can be rewritten using the law of total probability and \eqref{eq:app-proof-2}:
    \begin{align}    
    & \tilde{E}\left\{ \tilde{E}\left[\mu_{i^{\ast}}|T,\mathcal{O}\right]|\mathcal{O}\right\} 
      -\tilde{E}\left\{ \tilde{E}\left[\mu_{i^{\ast}}|T,\mathcal{O}\right]|\mathcal{D}\right\} \notag \\
    &= \tilde{E}\left[\mu_{i^{\ast}}|B,\mathcal{O}\right]
       + \tilde{P}\left(G|\mathcal{O}\right) 
         \left\{ \tilde{E}\left[\mu_{i^{\ast}}|G,\mathcal{O}\right]
         -\tilde{E}\left[\mu_{i^{\ast}}|B,\mathcal{O}\right]\right\} \notag \\
    &\quad -\tilde{E}\left[\mu_{i^{\ast}}|B,\mathcal{O}\right]
           -\tilde{P}\left(G|\mathcal{D}\right) 
           \left\{ 
            \tilde{E}\left[\mu_{i^{\ast}}|G,\mathcal{O}\right]
            -\tilde{E}\left[\mu_{i^{\ast}}|B,\mathcal{O}\right]
            \right\} \notag \\
    &= \left[\tilde{P}\left(G|\mathcal{O}\right)-\tilde{P}\left(G|\mathcal{D}\right)\right]
        \left\{ \tilde{E}\left[\mu_{i^{\ast}}|G,\mathcal{O}\right]
               -\tilde{E}\left[\mu_{i^{\ast}}|B,\mathcal{O}\right]\right\} 
       , \label{eq:optimal_post_hoc_2}
    \end{align}
    and the last line is $\left[\text{Darwinian Learning}\right]$.
\end{proof}

\end{appendices}

\newpage

\printbibliography

\end{document}